\documentclass[12pt]{article}
\usepackage{enumerate}
\usepackage{amsfonts}
\usepackage{amsmath}
\usepackage{amssymb}
\usepackage{amsthm}

\def\H{\mathcal{H}}

\def\M{\mathcal{M}}

\def\S{\mathfrak{S}}

\def\T{\mathfrak{T}}
\def\B{\mathfrak{B}}

\newcommand{\rank}{\mathrm{rank}}

\newcommand{\Tr}{\mathrm{Tr}}

\newcounter{defin}  \newcounter{lemma}  \newcounter{theorem}
\newcounter{property} \newcounter{corol}  \newcounter{remark} \newcounter{example}

\newenvironment{lemma}{\par\refstepcounter{lemma}
     \textbf{Lemma \thelemma.} }{\rm\par}
\newenvironment{theorem}{\par\refstepcounter{theorem}
     \textbf{Theorem \thetheorem.}\ }{\rm\par}
\newenvironment{property}{\par\refstepcounter{property}
     \textbf{Proposition \theproperty.}\ }{\rm\par}
\newenvironment{corollary}{\par\refstepcounter{corol}
     \textbf{Corollary \thecorol.} }{\rm\par}
\newenvironment{definition}{\par\refstepcounter{defin}
     \textbf{Definition \thedefin.}\ }{\rm\par}
\newenvironment{remark}{\par\refstepcounter{remark}
     \textbf{Remark \theremark.}}{\rm\par}

\begin{document}
\title{On superactivation of one-shot zero-error quantum capacity and the related property of quantum measurements}
\author{M.E. Shirokov\footnote{Steklov Mathematical Institute, RAS, Moscow, email:msh@mi.ras.ru}, T.V.
Shulman\footnote{University of Copenhagen, Denmark, email:shulman@math.ku.dk}}
\date{}
\maketitle

\begin{abstract} We begin with a detailed description of a low dimensional quantum channel
($d_A=4, d_E=3$) demonstrating the symmetric form of superactivation
of one-shot zero-error quantum capacity. This means appearance of a
noiseless (perfectly reversible) subchannel  in the tensor square of
a channel having no noiseless subchannels.

Then we describe a quantum channel $\Phi$ such that
$\,\bar{Q}_0(\Phi)=0$ and $\,\bar{Q}_0(\Phi\otimes\Phi)\geq\log n\,$
for any $\,n\leq+\infty$.

We also show that the superactivation of one-shot zero-error quantum
capacity of a channel can be reformulated in terms of quantum
measurements theory as  appearance of an indistinguishable subspace
for tensor product of two observables  having no undistinguishable
subspaces.
\end{abstract}
\maketitle


\section{Introduction}

The phenomenon of superactivation of quantum channel capacities has
been intensively studied since 2008 when G.Smith and J.Yard
established this property for the case of quantum capacity
\cite{S&Y}.

This phenomenon means that the particular capacity of the tensor
product of two quantum channels may be positive despite the same
capacity of each of these channels is zero. During the last five years
it was shown that superactivation holds for different quantum
channel capacities, in particular, for (one-shot and asymptotic)
zero-error classical and quantum capacities \cite{CCH,C&S,Duan}.

In this paper we focus attention on the superactivation of one-shot
zero-error quantum capacity which means that
\begin{equation}\label{sa-qc}
    \bar{Q}_0(\Phi_1)=\bar{Q}_0(\Phi_2)=0,\quad\textrm{but}\quad
    \bar{Q}_0(\Phi_1\otimes\Phi_2)>0
\end{equation}
for some channels $\Phi_1$ and $\Phi_2$, where $\bar{Q}_0$ denotes
the one-shot zero-error quantum capacity (described in Section 2).

This effect can be reformulated with no use the term "capacity" as
appearance of a noiseless (i.e. perfectly reversible) subchannel in
the tensor product of two channels each of which has no noiseless
subchannels. This reformulation seems more adequate for specialists
in functional analysis and operator algebras theory.

The existence of quantum channels, for which (\ref{sa-qc}) holds,
follows from the existence of quantum channels demonstrating so
called extreme superactivation of asymptotic zero-error capacities
shown in \cite{C&S} by rather inexplicit way in sufficiently high
dimensions. So, this result neither gives an explicit form of
channels demonstrating the superactivation of one-shot zero-error
quantum capacity, nor says anything about their minimal dimensions.

In our recent paper \cite{Sh&Sh} we explicitly describe low
dimensional channels $\Phi_1\neq\Phi_2$ ($\dim\H_A=8,\dim\H_E=5$)
demonstrating the extreme superactivation of one-shot zero-error
capacity which means (\ref{sa-qc}) with the condition
$\bar{Q}_0(\Phi_1)=\bar{Q}_0(\Phi_2)=0$ replaced by the stronger
condition $\bar{C}_0(\Phi_1)=\bar{C}_0(\Phi_2)=0$ (where $\bar{C}_0$
is the one-shot zero-error classical capacity). For these channels
superactivation (\ref{sa-qc}) obviously holds.

In this paper we use the same approach to construct \emph{more
simple} example of superactivation (\ref{sa-qc}). It turns out that
the change
$$
\bar{C}_0(\Phi_1)=\bar{C}_0(\Phi_2)=0\quad\rightarrow\quad\bar{Q}_0(\Phi_1)=\bar{Q}_0(\Phi_2)=0
$$
of prerequisites makes it possible to essentially decrease
dimensions ($\dim\H_A=4,\dim\H_E=3$) and to construct a symmetrical
example $\Phi_1=\Phi_2$, i.e. a such channel $\Phi$ that
\begin{equation*}
    \bar{Q}_0(\Phi)=0,\quad\textrm{but}\quad
    \bar{Q}_0(\Phi\otimes\Phi)>0.
\end{equation*}
Moreover, this channel $\Phi$ is defined via so simple
noncommutative graph, which gives possibility to write a minimal
Kraus  representation of $\Phi$ in explicit (numerical) form.

Then we describe a quantum channel $\Phi$ such that
\begin{equation*}
    \bar{Q}_0(\Phi)=0,\quad\textrm{but}\quad
    \bar{Q}_0(\Phi\otimes\Phi)\geq\log n,
\end{equation*}
where $\,n\,$ is any natural number or $\,+\infty\,$ (in the last
case $\Phi$ is an infinite-dimensional channel:
$\,\dim\H_A=\dim\H_B=+\infty$).\smallskip

In the last part of the paper (Section 3) we show  that the
superactivation of one-shot zero-error quantum capacity
(\ref{sa-qc}) has a counterpart in the theory of quantum
measurements. Namely, it can be reformulated as appearance of an
indistinguishable subspace for the tensor product of two quantum
observables having no indistinguishable subspaces. This observation
is quite simple but seems interesting for specialists in quantum
measurements theory.\smallskip

A general way to write the Kraus representation of a channel with
given noncommutative graph is considered in the Appendix.

\section{Superactivation of one-shot
zero-error quantum capacity}

Let $\H$ be a separable Hilbert space, $\B(\H)$ and
$\mathfrak{T}(\mathcal{H})$ -- the Banach spaces of all bounded
operators in $\mathcal{H}$ and of all trace-class operators in $\H$
correspondingly, $\S(\H)$ -- the closed convex subset of
$\mathfrak{T}( \H)$ consisting of positive operators with unit trace
called \emph{states} \cite{H-SCI,N&Ch}. If $\dim\H=n<+\infty$ we may
identify $\B(\H)$ and $\T(\H)$ with the space $\mathfrak{M}_n$ of
all $n\times n$ matrices (equipped with the appropriate norm).
\smallskip

Let
$\Phi:\mathfrak{T}(\mathcal{H}_A)\rightarrow\mathfrak{T}(\mathcal{H}_B)$
be a quantum channel, i.e. a  completely positive trace-preserving
linear map \cite{H-SCI,N&Ch}. Stinespring's theorem implies the
existence of a Hilbert space $\mathcal{H}_E$ and of an isometry
$V:\mathcal{H}_A\rightarrow\mathcal{H}_B\otimes\mathcal{H}_E$ such
that
\begin{equation}\label{Stinespring-rep}
\Phi(\rho)=\mathrm{Tr}_{\mathcal{H}_E}V\rho V^{*},\quad
\rho\in\mathfrak{T}(\mathcal{H}_A).
\end{equation}
The quantum  channel
\begin{equation}\label{c-channel}
\mathfrak{T}(\mathcal{H}_A)\ni
\rho\mapsto\widehat{\Phi}(\rho)=\mathrm{Tr}_{\mathcal{H}_B}V\rho
V^{*}\in\mathfrak{T}(\mathcal{H}_E)
\end{equation}
is called \emph{complementary} to the channel $\Phi$
\cite{H-SCI,H-c-c}. The complementary channel is defined uniquely up
to isometrical equivalence \cite[the Appendix]{H-c-c}.

The one-shot zero-error quantum capacity $\bar{Q}_0(\Phi)$ of a
channel $\Phi$ can be defined as $\;\sup_{\H\in
q_0(\Phi)}\log\dim\H\,$, where $q_0(\Phi)$ is the set of all
subspaces $\H_0$ of $\H_A$ on which the channel $\Phi$ is perfectly
reversible (in the sense that there is a channel $\Theta$ such that
$\Theta(\Phi(\rho))=\rho$ for all states $\rho$ supported by
$\H_0$). The (asymptotic) zero-error quantum capacity is defined by
regularization: $Q_0(\Phi)=\sup_n n^{-1}\bar{Q}_0(\Phi^{\otimes n})$
\cite{ZEC, CCH, C&S, Duan, W&Co}.

It is well known that a channel $\Phi$ is perfectly reversible on a
subspace $\H_0$ if and only if the restriction of the complementary
channel $\widehat{\Phi}$ to the subset $\S(\H_0)$ is completely
depolarizing , i.e. $\widehat{\Phi}(\rho_1)=\widehat{\Phi}(\rho_2)$
for all states $\rho_1$ and $\rho_2$ supported by $\H_0$
\cite[Ch.10]{H-SCI}. It follows that the one-shot zero-error quantum
capacity $\bar{Q}_0(\Phi)$ of a channel $\Phi$ is completely
determined by the set
$\mathcal{G}(\Phi)\doteq\widehat{\Phi}^*(\B(\H_E))$ called the
\emph{noncommutative graph} of $\Phi$ \cite{W&Co}. \smallskip

\begin{lemma}\label{trans-l+}
\emph{A channel $\,\Phi:\S(\H_A)\rightarrow\S(\H_B)$ is perfectly
reversible on the subspace $\,\H_0\subseteq\H_A$ spanned by the
family $\,\{\varphi_i\}_{i=1}^{n}$, $n\leq+\infty$, of orthogonal
unit vectors (which means that $\,\bar{Q}_0(\Phi)\geq\log n$) if and
only if
\begin{equation}\label{operators+}
\langle \varphi_i|A\varphi_j\rangle=0\quad\textit{and}\quad \langle
\varphi_i|A\varphi_i\rangle=\langle
\varphi_j|A\varphi_j\rangle\quad\forall i,j \;\;\forall
A\in\mathfrak{L},
\end{equation}
where $\mathfrak{L}=\mathcal{G}(\Phi)$ or, equivalently,
$\mathfrak{L}$ is any subset of $\,\B(\H_A)$ such that}
$$
\textit{the weak operator closure of}\;\;
\mathrm{lin}\mathfrak{L}\;=\;\textit{the weak operator closure
of}\;\;\mathcal{G}(\Phi).
$$
\end{lemma}

\begin{proof}
Relations (\ref{operators+}) mean that the complementary channel
$\widehat{\Phi}$ has completely depolarizing restriction to the
subset $\S(\H_0)$.
\end{proof}

\begin{remark}\label{trans-l+r}
Since a subspace $\mathfrak{L}$ of the algebra $\mathfrak{M}_n$ of
$n\times n$-matrices is a noncommutative graph of a particular
channel if and only if
\begin{equation}\label{L-cond}
\mathfrak{L}\;\,\textup{is
symmetric}\;\,(\mathfrak{L}=\mathfrak{L}^*)\;\,\textup{and contains
the unit matrix}
\end{equation}
(see Lemma 2 in \cite{Duan} or Proposition 2 in \cite{Sh&Sh}), Lemma
\ref{trans-l+} shows that one can "construct" a channel $\Phi$ with
$\dim\H_A=n$ having positive (correspondingly, zero) one-shot
zero-error quantum capacity by taking a subspace
$\mathfrak{L}\subset\mathfrak{M}_n$ satisfying (\ref{L-cond}) for
which the following condition is valid (correspondingly, not valid)
\begin{equation}\label{l-3-c}
\exists\varphi,\psi\in[\mathbb{C}^n]_1 \;\;\textup{s.t.}\;\; \langle
\psi|A\varphi\rangle=0\;\;\textup{and}\;\; \langle
\varphi|A\varphi\rangle=\langle \psi|A\psi\rangle\quad\forall
A\in\mathfrak{L},
\end{equation}
where $[\mathbb{C}^n]_1$ is the unit sphere of $\mathbb{C}^n$.

If $m$ is a natural number such that  $\dim\mathfrak{L}\leq m^2$,
then Corollary 1 in \cite{Sh&Sh} and Proposition \ref{cmp-c} in the
Appendix give explicit expressions of a channel $\Phi$ such that
$\mathcal{G}(\Phi)=\mathfrak{L}$ and $\dim\H_E\leq m$. $\square$
\end{remark}\medskip

Superactivation of one-shot zero-error quantum capacity means that
\begin{equation}\label{sa-cq+}
    \bar{Q}_0(\Phi_1)=\bar{Q}_0(\Phi_2)=0,\quad\textrm{but}\quad
    \bar{Q}_0(\Phi_1\otimes\Phi_2)>0.
\end{equation}
for some channels $\Phi_1$ and $\Phi_2$. As mentioned in the
Introduction the existence of channels $\Phi_1$ and $\Phi_2$ for
which (\ref{sa-cq+}) holds follows from the results in \cite{C&S},
but explicit examples of such channels with minimal dimensions are
not known (as far as we know).

Below we will construct a channel $\Phi$ with $\dim\H_A=4$,
$\dim\H_E=3$, $\dim\H_B=12$ such that (\ref{sa-cq+}) holds with
$\Phi_1=\Phi_2=\Phi$.\medskip

By Remark \ref{trans-l+r} the problem of finding channels, for which
(\ref{sa-cq+}) holds, is reduced to the problem of finding subspaces
$\mathfrak{L}_1$ and $\mathfrak{L}_2$ satisfying (\ref{L-cond}) such
that  condition (\ref{l-3-c}) is not valid  for
$\mathfrak{L}=\mathfrak{L}_1$ and for $\mathfrak{L}=\mathfrak{L}_2$
but it is valid for
$\mathfrak{L}=\mathfrak{L}_1\otimes\mathfrak{L}_2$. Now we will
consider a symmetrical example ($\mathfrak{L}_1=\mathfrak{L}_2$) of
such subspaces in $\mathfrak{M}_4$.

Let $U$ be the unitary operator in $\mathbb{C}^2$ determined (in the
canonical basis) by the matrix
$$
U =\left[\begin{array}{rr}
\eta &  0 \\
0  &  \bar{\eta}
\end{array}\right],
$$
where $\eta=\exp[\,\mathrm{i}\frac{\pi}{4}\,]$. Consider the
subspace
$$
\mathfrak{L}_0 = \left\{M=\left[\begin{array}{cc} A &
 \lambda U^*\\ \lambda U & A\end{array}\right],\;\; A\in \mathfrak{M}_2,\; \lambda\in\mathbb{C}\right\}
$$
of $\mathfrak{M}_4$. It obviously satisfies condition
(\ref{L-cond}).\medskip

\begin{theorem}\label{sqc} \emph{Condition
(\ref{l-3-c}) is not valid  for $\mathfrak{L}=\mathfrak{L}_0$ but it
is valid for  $\mathfrak{L}=\mathfrak{L}_0\otimes\mathfrak{L}_0$
with the vectors
\begin{equation}\label{vec+}
\!|\varphi_{t}\rangle=\textstyle\frac{1}{\sqrt{2}}\left[
|1\rangle\otimes|1\rangle+e^{\mathrm{i}t}|2\rangle\otimes|2\rangle\right]\!,\;|\psi_{t}\rangle=\textstyle\frac{1}{\sqrt{2}}\left[
|3\rangle\otimes|3\rangle+e^{\mathrm{i}t}|4\rangle\otimes|4\rangle\right]\!,\!\!
\end{equation}
where $\{|k\rangle\}_{k=1}^4$ is the canonical basis in
$\,\mathbb{C}^4$ and $\,t$ is a fixed number in $[0,2\pi)$.}
\end{theorem}

\begin{proof} Throughout the proof we will identify $\mathbb{C}^4$ with
$\mathbb{C}^2\oplus\mathbb{C}^2$.\medskip

Assume there exist unit vectors $\varphi=[x_1, x_2]$ and $\psi=[y_1,
y_2]$, $x_i,y_i\in\mathbb{C}^2$ such that $\langle\psi | M
\varphi\rangle=0$ and $\langle\psi | M \psi\rangle=\langle\varphi |
M  \varphi\rangle$ for all
   $M\in\mathfrak{L}_0$. It follows that
\begin{equation}\label{e-1}
    \langle y_1|Ax_1\rangle+\langle y_2|Ax_2\rangle=0\quad
\forall
   A\in\mathfrak{M}_2,
\end{equation}
\begin{equation}\label{e-2}
    \langle y_1|U^*x_2\rangle+\langle y_2|Ux_1\rangle=0,
\end{equation}
\begin{equation}\label{e-3}
    \langle y_1|Ay_1\rangle+\langle y_2|Ay_2\rangle=\langle x_1|Ax_1\rangle+\langle x_2|Ax_2\rangle\quad
\forall
   A\in\mathfrak{M}_2
\end{equation}
and
\begin{equation}\label{e-4}
    \langle y_1|U^*y_2\rangle+\langle y_2|Uy_1\rangle=\langle x_1|U^*x_2\rangle+\langle x_2|Ux_1\rangle.
\end{equation}

If $x_1\nparallel x_2$ then, by 2-transitivity of $\mathfrak{M}_2$,
there is $A_0\in\mathfrak{M}_2$ such that $y_1=A_0x_1$ and
$y_2=A_0x_2$ \cite{DMR}. So, (\ref{e-1}) implies $\langle
y_1|y_1\rangle+\langle y_2|y_2\rangle=0$, i.e. $y_1=y_2=0$.
Similarly, if $y_1\nparallel y_2$ then (\ref{e-1}) implies
$x_1=x_2=0$.

Thus, we necessarily have $x_1\parallel x_2$ and $y_1\parallel y_2$.
Now we will obtain a contradiction to (\ref{e-1})-(\ref{e-4}) by
considering the following cases.

Case 1: $x_2=0,\, x_1\neq0$. In this case (\ref{e-1}) implies
$\langle y_1|Ax_1\rangle=0$ for all $A\in\mathfrak{M}_2$, which can
be valid only if $y_1=0$. Then (\ref{e-3}) implies $\langle
x_1|Ax_1\rangle=\langle y_2|Ay_2\rangle$ for all
$A\in\mathfrak{M}_2$, which can be valid only if $x_1\parallel y_2$.
By Lemma \ref{sl} below this and (\ref{e-2}) show that $y_2=0$. So,
we obtain $y_1=y_2=0$.

Case 2: $y_2=0,\, y_1\neq0$. Similar to Case 1 we obtain
$x_1=x_2=0$.

Case 3: $x_2\neq0,\, y_2\neq0$. In this case $x_1=\mu x_2,\;y_1=\nu
y_2$ and (\ref{e-3}) implies
$$
    (1+|\mu|^2)\langle x_2|Ax_2\rangle=(1+|\nu|^2)\langle y_2|Ay_2\rangle\quad
\forall
   A\in\mathfrak{M}_2,
$$
which can be valid  only if $\,x_2\parallel y_2$. Hence we have
$x_1=\alpha y_2$ and $x_2=\beta y_2$ (in addition to $y_1=\nu y_2$).
We may assume that $x_1\neq0$ and $y_1\neq0$, since otherwise
(\ref{e-1}) implies $\langle y_2|Ax_2\rangle=0$ for all
$A\in\mathfrak{M}_2$, which can be valid only if either $x_2=0$ or
$y_2=0$.

It follows from (\ref{e-1}) that $(\bar{\nu}\alpha+\beta)\langle
y_2|y_2\rangle=0$ and hence
\begin{equation}\label{beta}
\beta=-\bar{\nu}\alpha.
\end{equation}
By the below Lemma \ref{sl} $z_0=\langle y_2|Uy_2\rangle$ is a
nonzero complex number. So,  (\ref{e-4}) and (\ref{beta}) imply
$\,\mathrm{Re}(\nu z_0)=\mathrm{Re}(\alpha\bar{\beta}
z_0)=-|\alpha|^2\mathrm{Re}(\nu z_0)\,$ and hence
\begin{equation}\label{re-eq}
\mathrm{Re}(\nu z_0)=0.
\end{equation}
It follows from (\ref{e-2}) and (\ref{beta}) that
$$
\bar{\nu}\beta \bar{z}_0+\alpha
z_0=\alpha(-\bar{\nu}^2\bar{z}_0+z_0)=0.
$$
Since $\alpha\neq0$ ($x_1\neq0$) we have $\nu^2z_0=\bar{z}_0$. This
equality implies that $\,\nu z_0\,$ is a real number. So,
(\ref{re-eq}) shows that $\,\nu=0\,$ contradicting to
$\,y_1\neq0$.\smallskip

Thus, condition (\ref{l-3-c}) is not valid for
$\mathfrak{L}=\mathfrak{L}_0$.\smallskip

Now we will show that
\begin{equation}\label{one}
    \langle\psi_{t} | M_1\otimes M_2 \, \varphi_{t}\rangle=0\quad \forall
    M_1,M_2\in\mathfrak{L}_0,
\end{equation}
and
\begin{equation}\label{two}
    \langle\psi_{t} | M_1\otimes M_2 \,\psi_{t}\rangle=\langle\varphi_{t} | M_1\otimes M_2 \, \varphi_{t}\rangle\quad \forall
    M_1,M_2\in\mathfrak{L}_0,
\end{equation}
where $\varphi_{t}$ and $\psi_{t}$ are vectors defined in
(\ref{vec+}). Since we identify $\mathbb{C}^4$ with
$\mathbb{C}^2\oplus\mathbb{C}^2$, these vectors are represented as
follows
\begin{equation*}
\begin{array}{l}
|\varphi_{t}\rangle=\textstyle{\frac{1}{\sqrt{2}}}\left[\;|e_1,0\rangle\otimes
|e_1,0\rangle+e^{\mathrm{i}t}|e_2,0\rangle\otimes
|e_2,0\rangle\,\right]\;\\\\
|\psi_{t}\rangle=\textstyle{\frac{1}{\sqrt{2}}}\left[\;|0,e_1\rangle\otimes
|0,e_1\rangle+e^{\mathrm{i}t}|0,e_2\rangle\otimes
|0,e_2\rangle\,\right]\!,\vspace{5pt}
\end{array}
\end{equation*}
where $\{|e_i\rangle\}$ is the canonical basis in
$\mathbb{C}^2$.\smallskip

By setting $\alpha_1=1$ and $\alpha_2=e^{\mathrm{i}t}$ we have
\begin{equation}\label{int}
M_1\otimes M_2 |\varphi_{t}\rangle=\frac{1}{\sqrt{2}}\sum_{i=1}^2
\alpha_i |A_1 e_i, \lambda_1U e_i\rangle\otimes|A_2 e_i, \lambda_2U
e_i\rangle,
\end{equation}
and hence
$$
\begin{array}{c}
\displaystyle\langle\psi_{t}| M_1\otimes M_2
\,\varphi_{t}\rangle=\frac{1}{2}\sum_{i,j=1}^2
\bar{\alpha}_i\alpha_j \langle 0, e_i|\otimes \langle 0, e_i|\cdot|A_1 e_j, \lambda_1U e_j\rangle\otimes|A_2 e_j, \lambda_2U e_j\rangle\\
\displaystyle=\frac{1}{2}\lambda_1\lambda_2 \sum_{i,j=1}^2
\bar{\alpha}_i\alpha_j \langle e_i |Ue_j\rangle\langle e_i
|Ue_j\rangle=\frac{1}{2}\,\lambda_1\lambda_2\left[\,\eta^2|\alpha_1|^2+\bar{\eta}^2|\alpha_2|^2\,\right]=0,\quad
\end{array}
$$
Thus (\ref{one}) is valid. It follows from (\ref{int}) that
\begin{equation}\label{two+}
\begin{array}{l}
\!\!\!\!\!\displaystyle\langle\varphi_{t}| M_1\otimes M_2\,
\varphi_{t}\rangle \qquad\qquad\qquad\qquad\qquad\qquad\qquad\qquad\qquad\qquad\qquad\\
\qquad\qquad\displaystyle=\frac{1}{2}\sum_{i,j=1}^2
\bar{\alpha}_i\alpha_j \langle e_i,0|\otimes \langle e_i,0 |\cdot|A_1 e_j, \lambda_1Ue_j\rangle\otimes|A_2 e_j, \lambda_2Ue_j\rangle\\
\qquad\qquad\displaystyle=\frac{1}{2}\sum_{i,j=1}^2
\bar{\alpha}_i\alpha_j
 \langle e_i | A_1  e_j\rangle\langle e_i |A_2
e_j\rangle.
\end{array}
\end{equation}
Since
\begin{equation*}
M_1\otimes M_2|\psi_{t}\rangle=\frac{1}{\sqrt{2}}\sum_{i=1}^2
\alpha_i|\lambda_1 U^*e_i, A_1 e_i\rangle\otimes|\lambda_2 U^*e_i,
A_2e_i\rangle
\end{equation*}
we have
\begin{equation*}
\begin{array}{c}
\displaystyle\langle\psi_{t}| M_1\otimes M_2\,
\psi_{t}\rangle=\frac{1}{2}\sum_{i,j=1}^2
\bar{\alpha}_i\alpha_j\langle 0,e_i|\otimes \langle 0,
e_i|\cdot|\lambda_1 U^*e_j,
A_1e_j\rangle\otimes|\lambda_2 U^*e_j, A_2e_j\rangle\\
\displaystyle=\frac{1}{2}\sum_{i,j=1}^2
\bar{\alpha}_i\alpha_j\langle e_i | A_1  e_j\rangle\langle e_i |A_2
 e_j\rangle.\qquad\qquad\;\;\;
\end{array}
\end{equation*}
This equality and (\ref{two+}) imply (\ref{two}).
\end{proof}

\begin{lemma}\label{sl}
\emph{If  $\,y$ is a nonzero vector in $\,\mathbb{C}^2$  then
$\,\langle y|Uy\rangle\neq 0$.}
\end{lemma}

\begin{proof}
Let $y=[y_1,y_2]$ then $Uy=[\eta y_1,\bar{\eta}y_2]$ and $\langle
y|Uy\rangle=|y_1|^2\eta+|y_2|^2\bar{\eta}\neq0$ (since
$\eta=\exp[\,\mathrm{i}\frac{\pi}{4}\,]$).
\end{proof}

By Proposition 2 in \cite{Sh&Sh} Theorem \ref{sqc-c} implies the
following assertion.\smallskip

\begin{corollary}\label{sqc-c} \emph{There is a pseudo-diagonal \footnote{A channel $\Phi:\S(\H_A)\rightarrow\S(\H_B)$ is called pseudo-diagonal
    if it has the representation
$$
\Phi(\rho)=\sum_{i,j}c_{ij}\langle
\psi_i|\rho|\psi_j\rangle|i\rangle\langle j|,\quad\rho\in\S(\H_A),
$$
where $\{c_{ij}\}$ is a Gram matrix of a collection of unit vectors,
$\{|\psi_i\rangle\}$ is a collection of vectors in $\H_A$ such that
$\;\sum_i |\psi_i\rangle\langle \psi_i|=I_{\H_A}\,$  and
$\{|i\rangle\}$ is an orthonormal basis in $\H_B$ \cite{R}.} channel
$\,\Phi$ with $\dim\H_A=4$, $\dim\H_E=3$, $\dim\H_B=12$ such that
$\,\mathcal{G}(\Phi)=\mathfrak{L}_0$ and hence
\begin{equation*}
    \bar{Q}_0(\Phi)=0,\quad\textrm{but}\quad
    \bar{Q}_0(\Phi\otimes\Phi)>0.
\end{equation*}}
\emph{The channel $\,\Phi\otimes\Phi$ is perfectly reversible on the
subspace
$\H_{t}=\mathrm{lin}\{|\varphi_{t}\rangle,|\psi_{t}\rangle\}$, where
$\varphi_{t},\psi_{t}$ are vectors defined in (\ref{vec+}), for each
given $\,t\in[0,2\pi)$.}
\end{corollary}\medskip

\begin{remark}\label{sqc-r} It is easy to see that the above subspace $\mathfrak{L}_0$ is not
transitive. So, by Lemma 2 in \cite{Sh&Sh}, the corresponding
channel $\Phi$ has positive one-shot zero-error classical capacity
and hence this channel does not demonstrate the extreme
superactivation of one-shot zero-error capacity.
\end{remark}\medskip

To obtain a minimal Kraus representation of one of the channels
having properties stated in Corollary \ref{sqc-c} we have to find a
basis $\{A_i\}_{i=1}^5$ of $\mathfrak{L}_0$ such that $\,A_i\geq0$
for all $i$ and $\,\sum_{i=1}^5 A_i=I_4$. Such basis can be easily
found, for example
$$
A_1=\frac{1}{6}\left[\begin{array}{cccc}
1 &  0 & \bar{\eta} &  0\\
0 &  2 & 0 &  \eta\\
\eta &  0 & 1 &  0\\
0 &  \bar{\eta} & 0 &  2
\end{array}\right]\!\!,
A_2=\frac{1}{6}\left[\begin{array}{rrrr}
\!1 & \!0 & \!-\bar{\eta} &  \!0\\
\!0 & \!2 & \!0 &  \!-\eta\\
\!-\eta & \!0 & \!1 &  \!0 \\
\!0 & \!-\bar{\eta} & \!0 &  \!2
\end{array}\right]\!\!,
A_3=\frac{5}{9}\left[\begin{array}{cccc}
1 & 0 & 0 & 0\\
0 & 0 & 0 & 0\\
0 & 0 & 1 & 0\\
0 & 0 & 0 & 0
\end{array}\right]\!\!,
$$
$$
A_4=\frac{1}{18}\left[\begin{array}{cccc}
1 &  \sqrt{3} & 0 &  0\\
\sqrt{3} &  3 & 0 &  0\\
0 &  0 & 1 &  \sqrt{3}\\
0 &  0 &  \sqrt{3} &  3
\end{array}\right]\!\!,
A_5=\frac{1}{18}\left[\begin{array}{cccc}
1 &  -\sqrt{3} & 0 &  0\\
-\sqrt{3} &  3 & 0 &  0\\
0 &  0 & 1 &  -\sqrt{3}\\
0 &  0 & -\sqrt{3} &  3
\end{array}\right]\!\!.
$$

We also have to chose a collection $\,\{|\psi_i\rangle\}_{i=1}^5$ of
unit vectors in $\mathbb{C}^3$ such that
$\,\{|\psi_i\rangle\langle\psi_i|\}_{i=1}^5$ is a linearly
independent subset of $\,\mathfrak{M}_3$. Let
$$
|\psi_1\rangle=|1\rangle,\;\, |\psi_2\rangle=|2\rangle,\;\,
|\psi_3\rangle=|3\rangle,\;\,
|\psi_4\rangle=\textstyle{\frac{1}{\sqrt{2}}}|1+3\rangle,\;\,
|\psi_5\rangle=\textstyle{\frac{1}{\sqrt{2}}}|2+3\rangle,
$$
where $\{|1\rangle,|2\rangle,|3\rangle\}$ is the canonical basis in
$\mathbb{C}^3$.

Now, by noting that $r_i=\rank A_i=3$ for $i=1,2$ and $r_i=\rank
A_i=2$ for $i=3,4,5$, we can apply Proposition \ref{cmp-c} in the
Appendix to obtain a minimal Kraus representation for
pseudo-diagonal channel $\Phi$ having properties stated in Corollary
\ref{sqc-c}. Direct calculation gives the following Kraus operators
$$
V_1=\frac{1}{6}\left[\begin{array}{cccc}
\sqrt{6} &  0 & \sqrt{6}\bar{\eta} &  0\\
0 &  \alpha & 0 &  \beta\\
0 &  \bar{\beta} & 0 &  \alpha\\
0 &  0 & 0 &  0\\
0 &  0 & 0 &  0\\
0 &  0 & 0 &  0\\
0 &  0 & 0 &  0\\
0 &  0 & 0 &  0\\
1 &  \sqrt{3} & 0 &  0\\
0 &  0 & 1 &  \sqrt{3}\\
0 &  0 & 0 &  0\\
0 &  0 & 0 &  0
\end{array}\right]\!\!,\quad
V_2=\frac{1}{6}\left[\begin{array}{cccc}
0 &  0 & 0 &  0\\
0 &  0 & 0 &  0\\
0 &  0 & 0 &  0\\
\sqrt{6} &  0 & -\sqrt{6}\bar{\eta} &  0\\
0 &  \alpha & 0 &  -\beta\\
0 &  -\bar{\beta} & 0 &  \alpha\\
0 &  0 & 0 &  0\\
0 &  0 & 0 &  0\\
0 &  0 & 0 &  0\\
0 &  0 & 0 &  0\\
1 &  -\sqrt{3} & 0 &  0\\
0 &  0 & 1 &  -\sqrt{3}
\end{array}\right],
$$
$$
V_3=\frac{1}{6}\left[\begin{array}{cccc}
0 &  0 & 0 &  0\\
0 &  0 & 0 &  0\\
0 &  0 & 0 &  0\\
0 &  0 & 0 &  0\\
0 &  0 & 0 &  0\\
0 &  0 & 0 &  0\\
2\sqrt{5} &  0 & 0 &  0\\
0 &  0 & 2\sqrt{5} &  0\\
1 &  \sqrt{3} & 0 &  0\\
0 &  0 & 1 &  \sqrt{3}\\
1 &  -\sqrt{3} & 0 &  0\\
0 &  0 & 1 &  -\sqrt{3}
\end{array}\right],\medskip
$$
where $\alpha=\displaystyle\frac{3+\sqrt{3}}{\sqrt{2}}$ and
$\beta=\displaystyle\eta\,\frac{3-\sqrt{3}}{\sqrt{2}}$
$\left(\,\eta=e^{\mathrm{i}\frac{\pi}{4}}\right)$. Thus,
$\,\Phi(\rho)=\displaystyle\sum_{k=1}^3 V_k\rho V_k^*$.

\section{Superactivation with  $\,\bar{Q}_0(\Phi\otimes\Phi)\geq\log n$}

By generalizing the above construction one can obtain the following
result.\smallskip

\begin{theorem}\label{sqc+} \emph{Let
$\,\dim\H_A=2n\leq+\infty$, $\,\{|k\rangle\}_{k=1}^{2n}$ an
orthonormal basis in $\H_A$, and $\,m\,$ the minimal natural number
such that $\;n^2-n+4\,\leq\, m^2$ if $\,n<+\infty\,$ and
$\,m=+\infty\,$ otherwise.}

\emph{There exists a pseudo-diagonal channel
$\,\Phi:\S(\H_A)\rightarrow\S(\H_B)$ with $\,\dim\H_E=m$ such that
$\,\bar{Q}_0(\Phi)=0$ while the channel $\,\Phi\otimes\Phi$ is
perfectly reversible on the subspace of $\,\H_A\otimes\H_A$ spanned
by the vectors
\begin{equation}\label{vec++}
|\varphi_k^{t}\rangle=\textstyle\frac{1}{\sqrt{2}}\left[
|2k-1\rangle\otimes|2k-1\rangle+e^{\mathrm{i}t}|2k\rangle\otimes|2k\rangle\right],\quad
k=1,2,\ldots,n,
\end{equation}
where $\,t$ is a fixed number in $\,[0,2\pi)$, and hence
$\,\bar{Q}_0(\Phi\otimes\Phi)\geq\log n$. }
\end{theorem}

\begin{proof} Assume first that $\,n<+\infty$. Consider the subspace
\begin{equation}\label{L_n}
\mathfrak{L}_n = \left\{\,M=\left[\begin{array}{ccccc}
 A              &  \lambda_{12} U^* &  \cdots & \lambda_{1n} U^*\\
 \lambda_{21} U &                 A &  \cdots & \lambda_{2n}
 U^*\\\cdots & \cdots & \cdots  & \cdots \\
\lambda_{n1} U &  \lambda_{n2} U   &  \cdots & A
 \end{array}\right],\;\; A\in \mathfrak{M}_2,\; \lambda_{ij}\in\mathbb{C}\;\right\}
\end{equation}
of $\mathfrak{M}_{2n}$, where  $U$ is the unitary operator in
$\mathbb{C}^2$ defined in the previous section (it has the matrix
$\,\mathrm{diag}\{\eta,\bar{\eta}\}\,$ in the canonical basis in
$\mathbb{C}^2$, $\,\eta=\exp[\,\mathrm{i}\frac{\pi}{4}\,]$).

The subspace $\mathfrak{L}_n$ satisfies condition (\ref{L-cond}) and
$\dim\mathfrak{L}_n=n^2-n+4$. So, by Proposition 2 in \cite{Sh&Sh},
there is a pseudo-diagonal channel $\Phi$ with $\dim\H_A=2n$ and
$\dim\H_E=m$, where $\,m\,$ is the minimal number satisfying the
inequality $\,n^2-n+4\,\leq\, m^2$, such that
$\mathcal{G}(\Phi)=\mathfrak{L}_n$.\smallskip

We will prove that $\bar{Q}_0(\Phi)=0\,$ by showing that condition
(\ref{l-3-c}) is not valid for $\mathfrak{L}=\mathfrak{L}_n$.

Assume there exist unit vectors $\varphi=[x_1, x_2,\ldots,x_n]$ and
$\psi=[y_1, y_2,\ldots,y_n]$, $x_i,y_i\in\mathbb{C}^2$, such that
$\langle\psi | M  \varphi\rangle=0$ and $\langle\psi | M
\psi\rangle=\langle\varphi | M  \varphi\rangle$ for all
   $M\in\mathfrak{L}_n$. It follows that
\begin{equation}\label{e-1+}
    \sum_{i=1}^n\langle y_i|Ax_i\rangle=0\quad\;
\forall
   A\in\mathfrak{M}_2,
\end{equation}
\begin{equation}\label{e-2+}
    \langle y_i|U^*x_k\rangle=0,\quad \forall\;  k >1,\;  i <
    k,
\end{equation}
\begin{equation}\label{e-2++}
    \langle y_i|Ux_k\rangle=0,\quad\;\, \forall\;  k < n,\;  i >
    k,
\end{equation}
and
\begin{equation}\label{e-3+}
    \sum_{i=1}^n\langle y_i|Ay_i\rangle=\sum_{i=1}^n\langle x_i|Ax_i\rangle\quad
\forall
   A\in\mathfrak{M}_2.
\end{equation}
Note that (\ref{e-3+}) means that
\begin{equation}\label{e-3++}
    \sum_{i=1}^n|y_i\rangle\langle y_i|=\sum_{i=1}^n|x_i\rangle\langle
    x_i|.
\end{equation}
It suffices to show that
\begin{equation}\label{collin}
\textup{either}\quad x_1\parallel x_2\parallel x_3\parallel\ldots
\parallel x_n\quad \textup{or}\quad y_1\parallel y_2\parallel
y_3\parallel\ldots \parallel y_n,
\end{equation}
since this and (\ref{e-3++}) imply $x_i\parallel y_j$ for all $i,j$,
which, by Lemma \ref{sl} in Section 2, contradicts to (\ref{e-2+})
and (\ref{e-2++}) (if $x_i=y_i=0$ for all $i\neq k$ then
 $\langle y_k|x_k\rangle=\langle \psi|\varphi\rangle=0$).

We will consider that the both vectors $\varphi$ and $\psi$ have at
least two nonzero components (since otherwise (\ref{collin})
obviously holds).

Let $k$ be the minimal number such that $x_i=y_i=0$ for all $i<k$
and either $x_k$ or $y_k$ is nonzero.

By symmetry we may assume that $x_k\neq0$. Then (\ref{e-2++})
implies
\begin{equation}\label{collin+}
y_{k+1}\parallel y_{k+2}\parallel\ldots \parallel y_n.
\end{equation}

If $\,y_k=0$ then this means (\ref{collin}). If $\,y_k\neq0$ then we
have the following three cases.

Case 1: $x_i\neq0$ and $y_j\neq0$, where $i>j>k$. In this case
(\ref{e-2+}) with $k=i$ shows that
\begin{equation*}
y_{k}\parallel y_{k+1}\parallel\ldots \parallel y_{i-1}.
\end{equation*}
Since $y_j\neq0$  and $\,i\geq k+2$, this and (\ref{collin+}) imply
(\ref{collin}).

Case 2:  $x_i\neq0$ and $y_j\neq0$, where $j>i>k$. Since $x_k\neq0$
and $y_k\neq0$, this case is reduced to the previous one by
permuting $\varphi$ and $\psi$.

Case 3: $x_i=y_i=0$ for all $i>k$ excepting $i=l>k$. In this case
(\ref{e-1+}) implies
$$
\langle y_k|Ax_k\rangle+\langle y_l|Ax_l\rangle=0\quad \forall
   A\in\mathfrak{M}_2.
$$
If $x_k\nparallel x_l$ then, by 2-transitivity of $\mathfrak{M}_2$,
there is $A_0\in\mathfrak{M}_2$ such that $y_k=A_0x_k$ and
$y_l=A_0x_l$ \cite{DMR}. So, the above equality  implies $\langle
y_k|y_k\rangle+\langle y_l|y_l\rangle=0$, which contradicts to the
assumption $y_k\neq0$. Thus $x_k\parallel x_l$ and (\ref{collin})
holds.

Now we will show that
\begin{equation}\label{one++}
    \langle\varphi_k^{t} | M_1\otimes M_2 \, \varphi_l^{t}\rangle=0\quad \forall
    M_1,M_2\in\mathfrak{L}_n,\;k\neq l
\end{equation}
and
\begin{equation}\label{two++}
    \langle\varphi_k^{t} | M_1\otimes M_2 \,\varphi_k^{t}\rangle=
    \langle\varphi_l^{t} | M_1\otimes M_2 \,\varphi_l^{t}\rangle\quad \forall
    M_1,M_2\in\mathfrak{L}_n,\;k\neq l,
\end{equation}
for the family $\{\varphi_k^{t}\}_{k=1}^{n}$ of vectors defined in
(\ref{vec++}). By Lemma \ref{trans-l+} these relations mean perfect
reversibility of the channel $\Phi\otimes\Phi$ on the subspace
spanned by this family, which implies
$\bar{Q}_0(\Phi\otimes\Phi)\geq\log n$.\smallskip

Let $|\xi^k_i\rangle=|0,\ldots,0,e_i,0,\ldots,0\rangle$ be a vector
in
$\mathbb{C}^{2n}=[\mathbb{C}^2\oplus\mathbb{C}^2\oplus\ldots\oplus
\mathbb{C}^2]$, where $e_i$ is in the $k$-th position ($\{e_1,e_2\}$
is the canonical basis in $\mathbb{C}^2$). Then
\begin{equation*}
|\varphi_k^{t}\rangle=\textstyle{\frac{1}{\sqrt{2}}}\left[\;|\xi^k_1\rangle\otimes
|\xi^k_1\rangle+e^{\mathrm{i}t}|\xi^k_2\rangle\otimes
|\xi^k_2\rangle\,\right],\quad k=1,2,...,n.
\end{equation*}
By setting $\alpha_1=1$ and $\alpha_2=e^{\mathrm{i}t}$, we have
\begin{equation}\label{int++}
M_1\otimes M_2 |\varphi_k^{t}\rangle=\frac{1}{\sqrt{2}}\sum_{j=1}^2
\alpha_j |\psi(1,k,j)\rangle\otimes |\psi(2,k,j)\rangle,
\end{equation}
where
$$
|\psi(r,k,j)\rangle=|\lambda^r_{1k}U^*e_j,\lambda^r_{2k}U^*e_j,\ldots,\lambda_{[k-1]k
}U^*e_j,
A^re_j,\lambda^r_{[k+1]k}Ue_j,\ldots,\lambda^r_{nk}Ue_j\rangle,
$$
$r=1,2$ ($A^r,\lambda^r_{ij}$ correspond to the matrix $M_r$). If
$\,l>k\,$ then
$$
\begin{array}{c}
\displaystyle\langle\varphi_l^{t}| M_1\otimes M_2
\,\varphi_k^{t}\rangle=\frac{1}{2}\sum_{i,j=1}^2
\bar{\alpha}_i\alpha_j \langle \xi^l_i|\otimes \langle \xi^l_i|\cdot|\psi(1,k,j)\rangle\otimes|\psi(2,k,j)\rangle\\
\displaystyle=\frac{1}{2}\lambda^1_{lk}\lambda^2_{lk} \sum_{i,j=1}^2
\bar{\alpha}_i\alpha_j \langle e_i |Ue_j\rangle\langle e_i
|Ue_j\rangle=\frac{1}{2}\,\lambda^1_{lk}\lambda^2_{lk}\left[\,\eta^2|\alpha_1|^2+\bar{\eta}^2|\alpha_2|^2\,\right]=0,\quad
\end{array}
$$
Thus (\ref{one++}) is valid for $l>k$ and hence for all $l\neq k$.
It follows from (\ref{int++}) that
\begin{equation}\label{two+++}
\begin{array}{c}
\displaystyle\langle\varphi_k^{t}| M_1\otimes M_2
\,\varphi_k^{t}\rangle=\frac{1}{2}\sum_{i,j=1}^2
\bar{\alpha}_i\alpha_j \langle \xi^k_i|\otimes \langle \xi^k_i|\cdot|\psi(1,k,j)\rangle\otimes|\psi(2,k,j)\rangle\\
\displaystyle=\frac{1}{2}\sum_{i,j=1}^2 \bar{\alpha}_i\alpha_j
 \langle e_i | A^1  e_j\rangle\langle e_i |A^2
e_j\rangle.
\end{array}
\end{equation}
and that
\begin{equation*}
\begin{array}{c}
\displaystyle\langle\varphi_l^{t}| M_1\otimes M_2
\,\varphi_l^{t}\rangle=\frac{1}{2}\sum_{i,j=1}^2
\bar{\alpha}_i\alpha_j\langle \xi^l_i|\otimes \langle
\xi^l_i|\cdot|\psi(1,l,j)\rangle\otimes|\psi(2,l,j)\rangle\\
\displaystyle=\frac{1}{2}\sum_{i,j=1}^2
\bar{\alpha}_i\alpha_j\langle e_i | A^1  e_j\rangle\langle e_i | A^2
 e_j\rangle.
\end{array}
\end{equation*}
This equality and
 (\ref{two+++}) imply (\ref{two++}).\medskip

Consider the case $n=+\infty$. Let $\H_A$ be a separable Hilbert
space represented as a countable direct sum of 2-D Hilbert spaces
$\mathbb{C}^2$. So, each operator in $\B(\H_A)$ can be identified
with infinite block matrix satisfying a particular "boundedness"
condition.

Let $\mathcal{L}_{*}$ be the set of all infinite block matrices $M$
defined in (\ref{L_n}) with $n=+\infty$ satisfying the condition
\begin{equation}\label{l-req}
    \Lambda^2=\sum_{i=1}^{+\infty}\sum_{j\neq i}|\lambda_{ij}|^2<+\infty.
\end{equation}
This condition guarantees boundedness of the corresponding operator
due to the following easily-derived inequality
\begin{equation}\label{2-ineq}
    \|M\|_{\B(\H_A)}^2\,\leq\,
    2\left[\,\|A\|_{\B(\mathbb{C}^2)}^2+\Lambda^2\,\right].
\end{equation}

Let $\overline{\mathcal{L}}_{*}$ be the operator norm closure of
$\mathcal{L}_{*}$. It is clear that $\overline{\mathcal{L}}_{*}$ is
a symmetric subspace of $\B(\H_A)$ containing the unit operator
$I_{\H_A}$. By using inequality (\ref{2-ineq}) it is easy to show
separability of the subspace $\overline{\mathcal{L}}_{*}$ in the
operator norm topology (as a countable dense subset of
$\overline{\mathcal{L}}_{*}$ one can take the set of all matrices
$M$ in which $A$ and all $\lambda_{ij}$ have rational components).

Symmetricity and separability of $\overline{\mathcal{L}}_{*}$ imply
(by the proof of Proposition 2 in \cite{Sh&Sh}) existence of a
countable subset
$\{\tilde{M}_i\}_{i=2}^{+\infty}\subset\overline{\mathcal{L}}_{*}$
of positive operators generating $\overline{\mathcal{L}}_{*}$ (i.e.
such that the operator norm closure of all linear combinations of
the operators $\tilde{M}_i$ coincides with
$\overline{\mathcal{L}}_{*}$). Let
$M_i=2^{-i}\|\tilde{M}_i\|^{-1}\tilde{M}_i$, $i=2,3,...$. Since
$I_{\H_A}\in\overline{\mathcal{L}}_{*}$ and the series
$\sum_{i=2}^{+\infty}M_i$ converges in the operator norm topology,
the positive operator $M_1=I_{\H_A}-\sum_{i=2}^{+\infty}M_i$ lies in
$\overline{\mathcal{L}}_{*}$. Thus, $\{M_i\}_{i=1}^{+\infty}$ is a
countable subset of $\overline{\mathcal{L}}_{*}\cap\B_{+}(\H_A)$
generating the subspace $\overline{\mathcal{L}}_{*}$ such that
\begin{equation}\label{ser}
 \sum_{i=1}^{+\infty}M_i=I_{\H_A},
\end{equation}
where the series converges in the operator norm topology.\smallskip

Let $\{|e_i\rangle\}_{i=1}^{+\infty}$ be an orthonormal basis in a
separable Hilbert space $\H_B$. Consider the unital completely
positive map
$$
\B(\H_B)\ni X\mapsto\Psi^*(X) = \sum_{i=1}^{+\infty}\langle e_i|X
e_i\rangle M_i\in\B(\H_A).
$$
Apparently all  $M_i$ lie in
$\mathrm{Ran}\Psi^*\doteq\Psi^*(\B(\H_B))$. Since the series in
(\ref{ser}) converges in the operator norm topology, $\mathrm{Ran}
\Psi^* \subseteq \,\overline{\mathcal{L}}_{*}$. Hence $\mathrm{Ran}
\Psi^*$ is a dense subset of $\,\overline{\mathcal{L}}_{*}$.

The predual map
$$
\T(\H_A)\ni \rho\mapsto\Psi(\rho) =
\sum_{i=1}^{+\infty}[\hspace{1pt}\Tr M_i\rho\,] |e_i\rangle\langle
e_i|\in\T(\H_B)
$$
is an entanglement-breaking quantum channel. Let $\Phi$ be the
complementary channel to $\Psi$, so that $\Phi$ is pseudo-diagonal
channel and $\mathcal{G}(\Phi)=\mathrm{Ran}\Psi^*$.

To prove that $\bar{Q}_0(\Phi)=0$ it suffices to show, by Lemma
\ref{trans-l+}, that condition (\ref{l-3-c}) is not valid for
$\mathfrak{L}=\mathfrak{L}_*$ (since $\mathfrak{L}_*$ and
$\mathrm{Ran}\Psi^*$ are dense in $\overline{\mathfrak{L}}_*$). This
can be done by repeating the arguments from the proof of the same
assertion in the case $n<+\infty$.\medskip

The vectors defined in (\ref{vec++}) with $n=+\infty$ are
represented as follows
\begin{equation*}
|\varphi_k^{t}\rangle=\textstyle{\frac{1}{\sqrt{2}}}\left[\;|\xi^k_1\rangle\otimes
|\xi^k_1\rangle+e^{\mathrm{i}t}|\xi^k_2\rangle\otimes
|\xi^k_2\rangle\,\right],\quad k=1,2,3,...,
\end{equation*}
where $|\xi^k_i\rangle=|0,\ldots,0,e_i,0,0,\ldots\rangle$ is a
vector in $\H_A=[\mathbb{C}^2\oplus\mathbb{C}^2\oplus\ldots\oplus
\mathbb{C}^2\oplus\ldots]$ containing $e_i$ in the $k$-th position
($\{e_1,e_2\}$ is the canonical basis in $\mathbb{C}^2$).\smallskip

Since $\mathrm{Ran} \Psi^*$ is a dense subset of
$\overline{\mathfrak{L}}_*$,
$\mathrm{Ran}\left[\Psi^*\otimes\Psi^*\right]$ is a dense subset of
$\overline{\mathfrak{L}}_*\bar{\otimes}\overline{\mathfrak{L}}_*$
(where $\bar{\otimes}$ denotes the spacial tensor product). So,  to
prove that the channel $\Phi\otimes\Phi$ is perfectly reversible on
the subspace  spanned by the family
$\{|\varphi_k^{t}\rangle\}_{k=1}^{+\infty}$ it suffices to show, by
Lemma \ref{trans-l+}, that that relations (\ref{operators+}) hold
for any pair $|\varphi_k^{t}\rangle,|\varphi_l^{t}\rangle$ and
$\mathfrak{L}=\{M_1\otimes M_2\,|\,M_1,M_2\in\mathfrak{L}_*\}$. This
can be done by the same way as in the proof of the similar relations
in the case $\,n<+\infty$.
\end{proof}

\section{One property of quantum measurements}

In this section we will show that the effect of superactivation of
one-shot zero-error quantum capacity has a counterpart in the theory
of quantum measurements.

In accordance with the basic postulates of quantum mechanics any
measurement of a quantum system associated with a Hilbert space $\H$
corresponds to a Positive Operator Valued Measure (POVM) also called
(generalized) \emph{quantum observable} \cite{H-SCI,N&Ch}. A quantum
observable with finite or countable set of outcomes is a discrete
resolution of the identity in $\B(\H)$, i.e. a set
$\{M_i\}_{i=1}^m$, $m\leq+\infty$, of positive operators in $\H$
such that $\sum_{i=1}^mM_i=I_{\H}$. An observable is called
\emph{sharp} if it corresponds to an orthogonal resolution of the
identity (in this case $\{M_i\}_{i=1}^m$ consists of mutually
orthogonal projectors).

If an observable  $\M=\{M_i\}_{i=1}^m$ is applied to a quantum
system in a given state $\rho$  then the probability of $i$-th
outcome is equal to $\Tr M_i\rho$. So, we may consider the
observable  $\M$ as the quantum-classical channel
$$
\S(\H)\ni\rho\mapsto\pi_{\M}(\rho)=\{\Tr
M_i\rho\}_{i=1}^m\in\mathfrak{P}_m,
$$
where $\mathfrak{P}_m$ is the set of all probability distributions
with $m$ outcomes.

In the theory of quantum measurements  the notion of
\emph{informational completeness} of an observable and its
modifications are widely used \cite{B,C&Co,P}. An observable $\M$ is
called informational complete if for any two different states
$\rho_1$ and $\rho_2$ the probability distributions
$\pi_{\M}(\rho_1)$ and $\pi_{\M}(\rho_2)$ are different.\smallskip

Informational non-completeness of an observable can be characterized
by the following notion.\footnote{We would be grateful for any
references concerning original definition of this notion.}\smallskip
\begin{definition}\label{und-s}
A subspace $\H_0\subset\H$ is called \emph{indistinguishable} for an
observable $\M$ if $\,\pi_{\M}(\rho_1)=\pi_{\M}(\rho_2)$ for any
states $\rho_1$ and $\rho_2$ supported by $\H_0$.
\end{definition}\smallskip

If $\M=\{M_i\}$ is a sharp observable then all its indistinguishable
subspaces coincide with the ranges of the projectors $M_i$ of rank
$\geq 2$. So, a sharp observable has no indistinguishable subspaces
if and only if it consists of one rank projectors. This is not true
for unsharp observables (see the example at the end of this
section).\smallskip

To describe indistinguishable subspaces of a given observable one
can use the following characterization of such subspaces.\smallskip
\begin{property}\label{und-s-c}
\emph{Let $\M=\{M_i\}_{i=1}^m$, $m\leq+\infty$, be an observable in
a Hilbert space $\H$ and $\H_0$  a subspace of $\,\H$. The following
statements are equivalent:}
\begin{enumerate}[(i)]
    \item \emph{$\H_0$ is an indistinguishable subspace for the observable $\M$;}
    \item \emph{$\langle\psi|M_i\varphi\rangle=0\,$ for all $\,i$ and any orthogonal
    vectors $\,\varphi,\psi\in\H_0$;}
    \item \emph{there exists  an orthonormal basis $\,\{|\varphi_k\rangle\}$ in
    $\H_0$ such that}
$$
\langle\varphi_k|M_i\varphi_j\rangle=0\quad and\quad
\langle\varphi_k|M_i\varphi_k\rangle=\langle\varphi_j|M_i\varphi_j\rangle\quad
\forall i,j,k.
$$
\end{enumerate}
\end{property}

\begin{proof} It suffices to note that the subspace $\H_0$
is  indistinguishable for the observable $\M$ if and only if the
quantum channel
\begin{equation}\label{m-ch}
\S(\H)\ni\rho\mapsto\sum_{i=1}^m[\Tr M_i\rho]|i\rangle\langle
i|\in\S(\H_m),
\end{equation}
where $\{|i\rangle\}$ is an orthonormal basis in the $m$-dimensional
Hilbert space $\H_m$, has completely depolarizing restriction to the
subset $\S(\H_0)\subset\S(\H)$ and to use the well known
characterizations of completely depolarizing channels.
\end{proof}

Nonexistence of indistinguishable subspaces for a quantum observable
can be treated as recognition quality of this observable. So, if we
have two observables $\M_1$ and $\M_2$ having no indistinguishable
subspaces it is natural to ask about the existence of
indistinguishable subspaces for their tensor product $\M_1\otimes
\M_2$.\footnote{If $\H_1$ and $\H_2$ are indistinguishable subspaces
for observables $\M_1$ and $\M_2$, then it is easy to see that
$\H_1\otimes\H_2$ is an indistinguishable subspaces for the
observable $\M_1\otimes \M_2$, but there is a possibility of
existence of \emph{entangled} indistinguishable subspaces for the
observable $\M_1\otimes \M_2$.} It turns out that this question is
closely related to the superactivation of one-shot zero-error
quantum capacity.\smallskip

\begin{property}\label{und-s-e}
\emph{Let $\,\H^1_A,\H^2_A$ be finite-dimensional Hilbert spaces. The following
statements are equivalent:}
\begin{enumerate}[(i)]
    \item \emph{there exist  channels $\Phi_1:\S(\H^1_A)\rightarrow\S(\H^1_B)$ and $\Phi_2:\S(\H^2_A)\rightarrow\S(\H^2_B)$
    with $\,\dim\mathcal{G}(\Phi_1)=m_1$ and
    $\,\dim\mathcal{G}(\Phi_2)=m_2$ such that}
\begin{equation*}
    \bar{Q}_0(\Phi_1)=\bar{Q}_0(\Phi_2)=0\quad\textit{and}\quad
    \bar{Q}_0(\Phi_1\otimes\Phi_2)\geq\log n;
\end{equation*}
    \item \emph{there exist observables $\M_1=\{M^1_i\}_{i=1}^{m_1}$ and $\M_2=\{M^2_i\}_{i=1}^{m_2}$ in spaces $\H^1_A$
    and $\H^2_A$  having no
indistinguishable subspaces such that the observable $\M_1\otimes
\M_2$ has a $\,n$-dimensional indistinguishable subspace.}\medskip

\emph{If $\,\Phi_1=\Phi_2$ in $\,\mathrm{(i)}$ then $\M_1=\M_2$ in
$\,\mathrm{(ii)}$ and vice versa.}
\end{enumerate}
\end{property}

\begin{proof} An
observable $\M=\{M_i\}_{i=1}^{m}$ has a $\,n$-dimensional
indistinguishable subspace if and only if the one-shot zero-error
quantum capacity of the channel complementary to channel
(\ref{m-ch}) is not less than $\log n$, this observable $\M$ has no
indistinguishable subspaces if and only if the above capacity is
zero. This follows from Lemma \ref{trans-l+} and Proposition
\ref{und-s-c}, since the output set of the channel dual to channel
(\ref{m-ch}) coincides with the subspace of $\B(\H_A)$ generated
 by the family $\{M_i\}_{i=1}^{m}$.\smallskip

$\mathrm{(ii)\Rightarrow(i)}$. This directly follows from the above
remark.\smallskip

$\mathrm{(i)\Rightarrow(ii)}$. By the proof of Proposition 2 in
\cite{Sh&Sh} there exist base $\{A^1_i\}_{i=1}^{m_1}$ and
$\{A^2_i\}_{i=1}^{m_2}$ of the subspaces $\mathcal{G}(\Phi_1)$ and
$\mathcal{G}(\Phi_2)$ consisting of positive operators such that
$\sum_{i=1}^{m_1}A^1_i=I_{\H^1_A}$ and
$\sum_{i=1}^{m_2}A^2_i=I_{\H^2_A}$. If we consider these base as
observables $\M_1$ and $\M_2$ then validity of $\mathrm{(ii)}$ can
be shown by using the remark at the begin of this proof.
\end{proof}

\begin{remark}\label{inf-dim}
By the above proof the implication $\mathrm{(ii)\Rightarrow(i)}$ in
Proposition \ref{und-s-e} holds for infinite-dimensional Hilbert
spaces $\,\H^1_A,\H^2_A$ and $n\leq\infty$. The implication
$\mathrm{(i)\Rightarrow(ii)}$ can be generalized to this case if the
noncommutative graphs $\,\mathcal{G}(\Phi_1),\mathcal{G}(\Phi_2)$ are separable.
This can be done by using the arguments at the end of the
proof of Theorem \ref{sqc+} instead of Proposition 2 in
\cite{Sh&Sh}.
\end{remark}\smallskip

Proposition \ref{und-s-e} and Corollary \ref{sqc-c} imply the
following result.\smallskip

\begin{corollary}\label{und-s-e-c}
\emph{There exists a quantum observable
$\mathcal{M}=\{M_i\}_{i=1}^{5}$ in $\,4\textup{-}\mathrm{D}$
Hilbert space with no indistinguishable subspaces such that the
observable $\M\otimes \M$ has a continuous family of $\,2\textup{-}\mathrm{D}$ indistinguishable subspaces.}
\end{corollary}\smallskip

As a concrete example of such observable $\M$ on can take the
resolution of the identity  $\{A_i\}_{i=1}^{5}$ described after
Corollary \ref{sqc-c} in Section 2. In this case each 2-D subspace
of $\mathbb{C}^4\otimes\mathbb{C}^4$ spanned by the vectors
(\ref{vec+}) is indistinguishable for $\M\otimes\M$.\smallskip

Proposition \ref{und-s-e} (with Remark \ref{inf-dim}) and Theorem
\ref{sqc+} imply the following observation.\smallskip

\begin{corollary}\label{und-s-e-c+}
\emph{Let $\,n\in\mathbb{N}\,$ or $\;n=+\infty$.  There exists a
quantum observable $\mathcal{M}=\{M_i\}_{i=1}^{n^2-n+4}$ in
$\,2n\textup{-}$dimensional Hilbert space with no indistinguishable
subspaces such that the observable $\M\otimes \M$ has a continuous family of
$\,n\textup{-}$dimensional indistinguishable subspaces.}\footnote{If
$\,n=+\infty$ then $\,n^2-n+4=+\infty$ and the
$\,n\textup{-}$dimensional Hilbert space (subspace) means a
separable Hilbert space (subspace).}
\end{corollary}\smallskip

\begin{remark}
The above effect of appearance of (entangled) indistinguishable
subspace for tensor product of two observables $\M_1$ and $\M_2$
having no indistinguishable subspaces \emph{does not hold for sharp
observables} $\M_1$ and $\M_2$ (since the tensor product of two
observables consisting of mutually orthogonal 1-rank projectors is
an observable consisting of mutually orthogonal 1-rank projectors as
well).
\end{remark}

\section*{Appendix: The Kraus representation of a channel with given noncommutative graph}

The following proposition is a modification of Corollary 1 in
\cite{Sh&Sh}.\smallskip

\begin{property}\label{cmp-c}
\emph{Let $\,\mathfrak{L}$ be a subspace of $\,\mathfrak{M}_n$,
$n\geq2$, satisfying condition (\ref{L-cond}) and $\{A_i\}_{i=1}^d$
a basis of $\,\mathfrak{L}$ such that $\,A_i\geq0$ for all $i$ and
$\,\sum_{i=1}^d A_i=I_n$.\footnote{The existence of a basis
$\{A_i\}_{i=1}^d$ with the stated properties for any subspace
$\mathfrak{L}$ satisfying condition (\ref{L-cond}) is shown in the
proof of Proposition 2 in \cite{Sh&Sh}.} Let $m$ be a natural number
such that $\,d=\dim\mathfrak{L}\leq m^2$ and
$\,\{|\psi_i\rangle\}_{i=1}^d$ a collection of unit vectors in
$\,\mathbb{C}^m$ such that
$\,\{|\psi_i\rangle\langle\psi_i|\}_{i=1}^d$ is a linearly
independent subset of $\,\mathfrak{M}_m$.}
\smallskip

\emph{For each $k=\overline{1,m}$ let $\,V_k$ be the linear operator
from $\H_A\doteq\mathbb{C}^n$ into
$\H_B\doteq\bigoplus_{i=1}^d\mathbb{C}^{r_i}$, where $r_i=\rank
A_i$, defined as follows
$$
V_k=\sum_{i=1}^d\langle k|\psi_i\rangle W_iA_i^{1/2},
$$
where $\{|k\rangle\}$ is the canonical basis in $\mathbb{C}^m$ and
$W_i$ is a partial isometry from $\H_A$ into $\H_B$ with the initial
subspace $\mathrm{Ran}A_i$ and the final subspace
$\mathbb{C}^{r_i}$. Then the channel
\begin{equation}\label{ch-phi}
\mathfrak{M}_n\ni\rho\mapsto\Phi(\rho)=\sum_{k=1}^mV_k\rho
V_k^*\in\mathfrak{M}_{r_1+...+r_d},
\end{equation}
is pseudo-diagonal and its noncommutative graph $\mathcal{G}(\Phi)$
coincides with $\,\mathfrak{L}$.}
\end{property}\medskip

\begin{proof}
In the proof of Corollary 1 in \cite{Sh&Sh} it is shown that the
channel
$$
\mathfrak{M}_n\ni\rho\mapsto\Psi(\rho) = \sum_{i=1}^d[\Tr A_i
\rho]|\psi_i\rangle\langle\psi_i|\in\mathfrak{M}_m
$$
has the Stinespring representation
$$
\Psi(\rho)=\Tr_{\mathbb{C}^n\otimes\mathbb{C}^d}V\rho V^*
$$
where
$$
V:|\varphi\rangle\mapsto\sum_{i=1}^dA_i^{1/2}|\varphi\rangle\otimes|i\rangle\otimes|\psi_i\rangle
$$
is an isometry from $\mathbb{C}^n$ into
$\mathbb{C}^n\otimes\mathbb{C}^d\otimes\mathbb{C}^m$ (here
$\{|i\rangle\}$ is the canonical basis in $\mathbb{C}^d$).

Since the channel $\Psi$ is entanglement-breaking and
$\Psi^*(\mathfrak{M}_m)=\mathfrak{L}$, its complementary channel
$$
\widehat{\Psi}(\rho)=\Tr_{\mathbb{C}^m}V\rho V^*
$$
is pseudo-diagonal and $\mathcal{G}(\widehat{\Psi})=\mathfrak{L}$.
Its Kraus representation is
$\widehat{\Psi}(\rho)=\sum_{k=1}^m\tilde{V}_k\rho \tilde{V}_k^*$,
where the operators $\tilde{V}_k$ are defined by the relation
$$
\langle\phi|\tilde{V}_k\varphi\rangle=\langle\phi\otimes
k|V\varphi\rangle,\quad\varphi\in\mathbb{C}^n,\phi\in\mathbb{C}^n\otimes\mathbb{C}^d,
$$
so that
$$
\tilde{V}_k|\varphi\rangle=\sum_{i=1}^d \langle k|\psi_i\rangle
A_i^{1/2}|\varphi\rangle\otimes|i\rangle.
$$
By identifying  $\mathbb{C}^n\otimes\mathbb{C}^d$ with
$\bigoplus_{i=1}^d\mathbb{C}^{n}$, it is easy to show that the
channel $\Phi$ defined by (\ref{ch-phi}) is isometrically equivalent
to the channel $\widehat{\Psi}$ (see \cite[the Appendix]{H-c-c}) and
hence $\mathcal{G}(\Phi)=\mathcal{G}(\widehat{\Psi})=\mathfrak{L}$.
\end{proof}

We are grateful to A.S.Holevo and to the participants of his seminar
"Quantum probability, statistic, information" (the Steklov
Mathematical Institute) for useful discussion.
\bigskip

The work of the first author is partially supported by RFBR grant
12-01-00319 and by the RAS research program. The work of the second
author is partially supported by the Danish
Research Council through the Centre for Symmetry and Deformation at the
University of Copenhagen.

\end{document}